\documentclass[british]{new_tlp}
\usepackage[T1]{fontenc}
\usepackage[latin9]{inputenc}
\usepackage{mathrsfs}
\usepackage{mathtools}
\usepackage{amsmath}
\let\proof\relax

\usepackage{amsthm}
\usepackage{amssymb}
\usepackage{stackrel}

\makeatletter
\numberwithin{equation}{section}
\numberwithin{figure}{section}
\theoremstyle{plain}
\newtheorem{thm}{\protect\theoremname}
\theoremstyle{definition}
\newtheorem{defn}[thm]{\protect\definitionname}
\theoremstyle{remark}
\newtheorem*{rem*}{\protect\remarkname}
\theoremstyle{plain}
\newtheorem{prop}[thm]{\protect\propositionname}
\theoremstyle{remark}
\newtheorem*{notation*}{\protect\notationname}
\theoremstyle{plain}
\newtheorem{fact}[thm]{\protect\factname}
\theoremstyle{plain}
\newtheorem{cor}[thm]{\protect\corollaryname}
\theoremstyle{remark}
\newtheorem{exam}{\protect\examplename}

\makeatother

\usepackage{babel}
\addto\captionsbritish{\renewcommand{\corollaryname}{Corollary}}
\addto\captionsbritish{\renewcommand{\definitionname}{Definition}}
\addto\captionsbritish{\renewcommand{\factname}{Fact}}
\addto\captionsbritish{\renewcommand{\notationname}{Notation}}
\addto\captionsbritish{\renewcommand{\propositionname}{Proposition}}
\addto\captionsbritish{\renewcommand{\remarkname}{Remark}}
\addto\captionsbritish{\renewcommand{\theoremname}{Theorem}}
\addto\captionsenglish{\renewcommand{\corollaryname}{Corollary}}
\addto\captionsenglish{\renewcommand{\definitionname}{Definition}}
\addto\captionsenglish{\renewcommand{\factname}{Fact}}
\addto\captionsenglish{\renewcommand{\notationname}{Notation}}
\addto\captionsenglish{\renewcommand{\propositionname}{Proposition}}
\addto\captionsenglish{\renewcommand{\remarkname}{Remark}}
\addto\captionsenglish{\renewcommand{\theoremname}{Theorem}}
\providecommand{\corollaryname}{Corollary}
\providecommand{\definitionname}{Definition}
\providecommand{\factname}{Fact}
\providecommand{\notationname}{Notation}
\providecommand{\propositionname}{Proposition}
\providecommand{\remarkname}{Remark}
\providecommand{\examplename}{Example}
\providecommand{\theoremname}{Theorem}

\begin{document}
  \title[Asymptotic analysis of probabilistic logic programming]
        {An asymptotic analysis of probabilistic logic programming, with implications for expressing projective families of distributions\thanks{We would like to thank Manfred Jaeger for his encouragement and for helpful conversations about the subject of this paper, and the anonymous reviewers for facilitating a clearer exposition of the material.}}

  \author[F. Q. Weitk\"amper]
         {FELIX Q. WEITK\"AMPER\\
         Institut f\"ur Informatik, LMU M\"unchen, Oettingenstr. 67, 80538 M\"unchen\\
         \email{felix.weitkaemper@lmu.de}}

\maketitle
\begin{abstract}
Probabilistic logic programming is a major part of statistical relational artificial intelligence, where approaches from logic and probability are brought together to reason about and learn from relational domains in a setting of uncertainty. However, the behaviour of statistical relational representations across variable domain sizes is complex, and scaling inference and learning to large domains remains a significant challenge.
In recent years, connections have emerged between domain size dependence, lifted inference and learning from sampled subpopulations. The asymptotic  behaviour  of statistical relational representations has come under scrutiny, and projectivity was investigated as the strongest form of domain-size dependence, in which query marginals are completely independent of the domain size. 

In this contribution we show that every probabilistic logic program  under the distribution semantics is asymptotically equivalent to an acyclic probabilistic logic program consisting only of determinate clauses over probabilistic facts.
We conclude that every probabilistic logic program inducing a projective family of distributions is in fact everywhere equivalent to a program from this fragment, and we investigate the consequences for the projective families of distributions expressible by probabilistic logic programs. 

To facilitate the application of classical results from finite model theory, we introduce the abstract distribution semantics, defined as an arbitrary logical theory over probabilistic facts. This bridges the gap to the distribution semantics underlying probabilistic logic programming. In this representation, determinate logic programs correspond to quantifier-free theories, making asymptotic quantifier elimination results available for the setting of probabilistic logic programming.

\end{abstract}

\section{Introduction: Projectivity and statistical relational artificial
intelligence}

Statistical relational artificial intelligence
has emerged over the last 25 years as a means to specify statistical
models for relational data. Since then, many different frameworks
have been developed under this heading, which can broadly be classified
into those who extend logic programming to incorporate probabilistic
information (probabilistic logic programming under the distribution
semantics) and those who specify an abstract template for probabilistic
graphical models (sometimes known as knowledge-based model construction). 

Both classes share the distinction between a general model (a template
or a probabilistic logic program with variables) and a specific domain
used to ground the model. Ideally, the model would be specified abstractly
and independently of a specific domain, even though a specific domain
may well have been involved in learning the model from data. 

However, a significant hurdle is the generally hard to predict or
undesirable behaviour of the model when
applied to domains of different sizes. This extrapolation problem has received
much attention in the past years \cite{PooleBKKN14,JaegerS20}.
Recently Jaeger
and Schulte \shortcite{JaegerS18,JaegerS20}
have identified \emph{projectivity} as a strong form of good scaling
behaviour: in a projective model, the probability of a given property
holding for a given object in the domain is completely independent
of the domain size. However, the examples of Poole et al. \shortcite{PooleBKKN14} show
that projectivity cannot be hoped for in general statistical relational
models, and Jaeger and Schulte \shortcite{JaegerS18} identify very restrictive fragments
of common statistical relational frameworks as projective.

The question remains, however, whether those fragments completely capture the projective families of distributions expressible by a statistical relational representation. 
We will show in this contribution that in the case of probabilistic
logic programming under the distribution semantics, this is true, as every projective
probabilistic logic program is equivalent to a determinate acyclic probabilistic
logic program. 

Our method will show that, moreover, every probabilistic
logic program is asymptotically equivalent to an acyclic determinate probabilistic
logic program. This result is of some independent interest, as it shows that the probabilities of queries expressed by a logic program converge as domain size increases. Moreover, the asymptotic equivalence provides an explicit representation using which the asymptotic query probabilities can be computed. 

This will be an application of an asymptotic quantifier elimination
result for probabilistic logic programming derived from classical
finite model theory, namely from the study of the asymptotic theory
of first-order and least fixed point logic in the 1980s (particularly 0-1 laws, applied in the form of Blass et al. \shortcite{BlassGK85}). 

This application is also methodologically interesting as it opens
another way in which classical logic can contribute to
cutting-edge problems in learning and reasoning. That the
theory developed around 0-1 laws would be a natural candidate for
such investigations may not surprise, as it is highly developed and
is itself in the spirit of ``finite probabilistic model theory''
\cite[Section 7]{CozmanM19}, and one might hope for more
cross-fertilisation between the two fields in the future.

\subsection{Outline of the paper}

We will first formally introduce the framework of families of distributions and the notion of projectivity that we will refer to throughout.  

In the following section, we present the abstract distribution semantics, which bridges the gap between the tools from finite model theory and the semantics of probabilistic logic programming.
We also discuss asymptotic quantifier elimination and introduce the main classical results from finite model theory. 

We introduce least fixed point logic, an adequate representation for (probabilistic) logic programs.
We then give the necessary background on the asymptotic behaviour of least fixed point logic. 
We harness the relationship between probabilistic logic programming and least fixed point distributions to show that every probabilistic logic program is asymptotically equivalent to an acyclic determinate probabilistic logic program.

In the following section, we will apply this analysis to study the projective families of distributions expressible in probabilistic logic programming. We see that every projective
logic program is actually everywhere equivalent to an acyclic determinate logic program, and we derive some properties for the projective distributions expressible in this way.
For the case of a unary vocabulary, we show that only very few projective families of distributions are expressible in probabilistic logic programming, and we give a concrete example to highlight that point.

Finally, we conclude the paper with a brief discussion of the complexity of asymptotic quantifier elimination and some impulses for further research.

Proofs to all the statements made here can be found in Appendix A in the supplementary material.

\subsection{Notation}
An introduction to the terminology of first-order logic used in this paper can be found in Appendix B.1, in the supplementary material.
We just point out here that we use $\mathfrak{P}(A)$ to
indicate the power set of a set $A$ and $\vec{x}$ as a shorthand for a finite tuple $x_1, \dots, x_n$ of arbitrary length. 
\subsection{Projectivity}

We will introduce projective families of distributions in accordance
with Jaeger and Schulte \shortcite{JaegerS18,JaegerS20}, where one can find a much more detailed
exposition of the terms and their background. As we are interested
in statistical relational representations as a means of abstracting
away from a given ground model, we will refer to families of distributions
with varying domain sizes.
\begin{defn}
A \emph{family of distributions} for a relational vocabulary $\mathcal{S}$
is a sequence $\left(Q^{(n)}\right)_{n\in\mathbb{N}}$ of probability
distributions on the sets $\Omega_{n}$ of all $\mathcal{S}$-structures
with domain $\left\{ 1,\ldots,n\right\} \subseteq\mathbb{N}$. 
\end{defn}

\begin{defn}
A family of distributions is called \emph{exchangeable }if every  $Q^{(n)}$ is invariant
under $\mathcal{S}$-isomorphism. 

It is called \emph{projective} if, in addition, for all $m<n\in\mathbb{N}$ and
all $\omega\in\Omega_{m}$ the following holds: 
\[
Q^{(m)}(\{\omega\})=Q^{(n)}\left(\left\{ \omega'\in\Omega_{n}|\textrm{\ensuremath{\omega} is the substructure of }\omega'\textrm{ with domain \ensuremath{\left\{  1,\ldots,n\right\} } }\right\} \right)
\]
 
\end{defn}

Projectivity encapsulates a strong form of domain size independence.
Consider, for instance, the query $R(x)$, where $R$ is a relation symbol in $\mathcal{S}$.
Then in an exchangeable family of distributions, the unconditional probability of $R(x)$ holding in a world is independent of the precise interpretation of $x$, and depends only on the domain size.
If the family of distributions is projective, then the probability of $R(x)$ is independent even of the domain size.  
As an immediate consequence, this implies that the computational complexity of quantifier-free queries is constant with domain size, since queries can always be evaluated in a domain consisting just of the terms mentioned in the query itself.
Projectivity also has important consequences for the statistical consistency of learning from randomly sampled subsets \cite{JaegerS18}.

An important class of examples of projective families of distributions are those in which $R(a)$ is independent of $P(b)$ for all $R,P,a,b$.
For instance, consider a vocabulary $\mathcal{S}$  with unary relations $P$ and $R$, and a family of distributions in which for every domain element $a$, $P(a)$ and $R(a)$ are determined independently with probabilities $p$ and $r$ respectively.
Then the probability that a subset $A$ of a domain $D$ has $\mathcal{S}$-structure $M$ is given by
  \[
  p^{\left| a \in M | M \models P(a) \right|} \cdot (1-p)^{\left| a \in M | M \models \neg P(a) \right|} \cdot r^{\left| a \in M | M \models R(a) \right|} \cdot (1-r)^{\left| a \in M | M \models \neg R(a) \right|}
  \]
regardless of the size of $D$.

The work of Jaeger and Schulte \shortcite{JaegerS20} provides a complete characterisation of projective families of distributions in terms of exchangeable arrays (AHK representations).
However, it is not clear how this representation translates to the statistical relational formalisms currently in use, such as probabilistic logic programming.
We will see below that there are indeed projective families of distributions that are not expressible by a probabilistic logic program.
Furthermore, Jaeger and Schulte \shortcite{JaegerS20} claimed in Proposition 7.1 of their paper an independence property for the subclass of AHK- distributions.
While this proposition proved to be incorrect for the class of AHK- distributions \cite{JaegerS20a}, we will see here that for a projective family of distributions induced by a probabilistic logic program, the independence property holds.

In the remainder of this paper, we will investigate the interplay
between the asymptotic behaviour of logical theories as they have
been studied in finite model theory and the families of distributions
that are induced by them. We therefore introduce a notion of asymptotic
equivalence of families of distributions.

\section{\label{sec:Asymptotic-Quantifier-Elimination}Abstract distribution semantics}

As a bridge between classical notions from finite model theory and probabilistic logic programming, we introduce the abstract distribution semantics.
It builds on the \emph{relational Bayesian network specifications} of Cozman and Maua \shortcite{CozmanM19}, which combine random and independent root predicates with non-root predicates that are defined by first-order formulas.
Here we streamline and generalise this idea to a unified framework that we call the \emph{abstract distribution semantics}.
In particular, we will generalise away from first-order logic (FOL) to a general \emph{logical language:}
\begin{defn}
Let $\mathcal{R}$ be a vocabulary. Then a \emph{logical language
}$L(\mathcal{R})$ consists of a collection of
functions $\varphi$  which take an $\mathcal{R}$-structure
$M$ and returns a subset of  $M^{n}$ for some $n\in\mathbb{N}$ (called
the \emph{arity} of $\varphi$).
In analogy to the formulas of first-order logic, we refer to those functions as \emph{$L(\mathcal{R})$-formulas} and write $M\models\varphi(\vec{a})$ whenever $\vec{a} \in  \varphi(M)$.
\end{defn}

The archetype of a logical language is the first-order predicate calculus, where an $R$-formula  $\varphi$ defines a function  $\varphi(M) := \{a \in M | M \models \varphi(a) \}$ and $\models$ is used in the sense of ordinary first-order logic.
The concept as defined here is sufficiently general to accommodate
many other choices, however, and we will later apply it to least fixed point logic in particular.

\begin{defn}
  Let $\mathcal{S}$ be a relational vocabulary,
  $\mathcal{R}\subseteq\mathcal{S}$, and let $L(\mathcal{R})$
be a logical language over $\mathcal{R}$. Then an \emph{abstract
$L$-distribution over $\mathcal{R}$ (with vocabulary $\mathcal{S}$)}
consists of the following data:

For every $R\in\mathcal{R}$ a number $q_{R}\in\mathbb{Q}\cap[0,1]$.

For every $R\in\mathcal{S}\backslash\mathcal{R}$, an $L(\mathcal{R})$-formula
$\phi_{R}$ of the same arity as $R$. 
\end{defn}

In the following we will assume that all vocabularies are finite.
The semantics of an abstract distribution is only defined relative
to a domain $D$, which we will also assume to be finite. The formal
definition is as follows:
\begin{defn}
Let $L(\mathcal{R})$ be a logical language over $\mathcal{R}$ and
let $D$ be a finite set. Let $T$ be an abstract $L$-distribution
over $\mathcal{R}$. Let $\Omega_{D}$ be the set of all $\mathcal{R}$-structures
with domain $D$.

Then the \emph{probability distribution on $\Omega_{D}$ induced by
$T$, written $Q_{T}^{(D)}$, }is defined as follows:

For all $\omega\in\Omega_{D}$, if $\exists_{\vec{a}\in\vec{D}}\exists_{R\in\mathcal{S} \setminus \mathcal{R}}:R(\vec{a})\nLeftrightarrow\phi_{R}(\vec{a})$,
then $Q_{T}^{(D)}(\{\omega\})\coloneqq0$

Otherwise, $Q_{T}^{(D)}(\{\omega\})\coloneqq\underset{R\in\mathcal{R}}{\prod}(q_{R}^{|\{\vec{a}\in\vec{D}|R(\vec{a})\}|})\times\underset{R\in\mathcal{R}}{\prod}(1-q_{R})^{|\{\vec{a}\in\vec{D}|\neg R(\vec{a})\}|}$
\end{defn}

In other words, all the relations in $\mathcal{R}$ are independent  
with probability $q_{R}$ and the relations in $\mathcal{S} \setminus \mathcal{R}$
are defined deterministically by the $L(\mathcal{R})$-formulas $\phi_{R}$.
We will illustrate that with an example.

\begin{exam}
  Let $\mathcal{R}=\{R,P\}$, $\mathcal{S} = \{R,P,S\}$, for  a unary relation $R$ a binary relation $P$ and a unary relation $S$.
  Then an abstract distribution over $\mathcal(R)$ has numbers $q_{R}$ and $q_{P}$ which encode probabilities. Consider the FOL-distribution $T$ with $\varphi_S = \exists_y \left( R(x) \wedge P(x,y) \right)$.
  For any domain $D$, $Q_{T}^{(D)}$ is obtained by making an independent choice of $R(a)$ or $\neg R(a)$ for every $a \in D$, with a  $q_{R}$ probability of $R(a)$. Similarly, an independent choice of $P(a,b)$ or $\neg P(a,b)$ is made for every pair $(a,b)$ from $D^2$, with a  $q_{P}$ probability of $P(a,b)$.
  Then, for any possible $\mathcal{R}$-structure, the interpretation of $S$ is determined by $\forall_x S(x) \leftrightarrow  \varphi_S(x)$.
  The resulting family of distributions is not projective, since the probability of $Q(a)$ increases with the size of the domain as more possible candidates $b$ for $P(a,b)$ are added. 

\end{exam}

As this example has shown, abstract FOL distributions do not necessarily give rise to projective families.
If the $\varphi$ are all given by quantifier-free formulas, however, then the induced families distributions are indeed projective. We call such abstract $L(\mathcal{R})$- distributions, in which $L(\mathcal{R})$ is the class of quantifier-free FOL-formulas over $\mathcal{R}$, \emph{quantifier-free distributions}.

\begin{prop}\label{prop:QF_implies_projective}
  Every abstract quantifier-free distribution induces a projective family of distributions.
\end{prop}

Quantifier-free distributions also hold a special role in model-theoretic analysis.
In particular, \emph {asymptotic quantifier elimination} has been shown for various logics of interest to artificial intelligence.

\subsection{Asymptotic quantifier elimination}
We introduce our notion of asymptotic equivalence for families of distributions:
\begin{defn}
Two families of distributions $(Q^{(n)})$ and $(Q'^{(n)})$ are \emph{asymptotically
equivalent} if $\stackrel[n\rightarrow\infty]{}{\lim}\underset{A\subseteq\Omega_{n}}{\sup}|Q^{(n)}(A)-Q'^{(n)}(A)|=0$
\end{defn}
\begin{rem*}
In measure theoretic terms, the families of distributions $(Q^{(n)})$
and $(Q'^{(n)})$ are asymptotically equivalent if and only if the
limit of the total variation difference between them is $0$. 
\end{rem*}
We extend the notion to abstract distributions by calling abstract distributions asymptotically equivalent if they induce asymptotically equivalent families of distributions. 
This gives us the following setting for asymptotic quantifier elimination:
\begin{defn}
  Let $L(\mathcal{R})$ be an extension of the class of quantifier-free $\mathcal{R}$-formulas.
  Then \emph{$L(\mathcal{R})$ has asymptotic quantifier elimination} if every abstract $L(\mathcal{R})$ distribution is asymptotically equivalent to a quantifier-free distribution over $L(\mathcal{R})$.
\end{defn}

It is well-known that first-order logic has asymptotic quantifier elimination. 

Indeed, the asymptotic theory of relational first-order logic can be summarised as follows \cite[Chapter 4]{EbbinghausF06}:
\begin{defn}
\label{fact:Random} Let $\mathcal{R}$ be a relational vocabulary.
Then the first order theory\emph{ }$\mathrm{RANDOM}(\mathcal{R})$
is given by all axioms of the following form, called \emph{extension
axioms over $\mathcal{R}$}: 

\[
\forall_{v_{1},\ldots,v_{r}}\left(\underset{1\leq i<j\leq r}{\bigwedge}v_{i}\neq v_{j}\rightarrow\exists_{v_{r+1}}\left(\underset{1\leq i\leq r}{\bigwedge}v_{i}\neq v_{r+1}\wedge\underset{\varphi\in\Phi}{\bigwedge}\varphi\wedge\underset{\varphi\in\Delta_{r+1}\backslash\Phi}{\bigwedge}\neg\varphi\right)\right)
\]
where $r\in\mathbb{N}$ and $\Phi$ is a subset of 
\[
\Delta_{r+1}\coloneqq\left\{ R(\vec{x})|R\in\mathcal{R},\textrm{ \ensuremath{\vec{x}} a tuple from \ensuremath{\{v_{1},\ldots,v_{r+1}\}} containing \ensuremath{v_{r+1}}}\right\} .
\]
 
\end{defn}

\begin{fact}
\label{fact:FOL_QE}$\mathrm{\ensuremath{RANDOM}(\mathcal{R})}$ eliminates
quantifiers, i. e. for each formula $\varphi(\vec{x})$ there is a
quantifier-free formula $\varphi'(\vec{x})$ such that $\mathrm{\ensuremath{RANDOM}(\mathcal{R})}\vdash\forall_{\vec{x}}(\varphi(\vec{x})\leftrightarrow\varphi'(\vec{x}))$.
\end{fact}

It is sometimes helpful to characterise this quantifier-free formula
somewhat more explicitly:
\begin{prop}
\label{prop:form_QE}Let $\varphi(\vec{x})$ be a formula of first-order
logic. Then:
\begin{enumerate}
\item $\varphi'(\vec{x})$ as in Fact \ref{fact:FOL_QE} can be chosen such
that only those relation symbols occur in $\varphi'$ that occur in
$\varphi$.
\item If every atomic subformula of $\varphi$ contains at least one free
variable not in $\vec{x}$, and no relation symbol occurs with different
variables in different literals, then either $\mathrm{\ensuremath{RANDOM}(\mathcal{R})}\vdash\forall_{\vec{x}}\varphi(\vec{x})$
or $\mathrm{\ensuremath{RANDOM}(\mathcal{R})}\vdash\forall_{\vec{x}}\neg\varphi(\vec{x})$.
\end{enumerate}
\end{prop}

The importance of $\mathrm{RANDOM}(\mathcal{R})$ comes from its role
as the asymptotic limit of the class of all $\mathcal{R}$-structures.
In fact, it axiomatises the limit theory of $\mathcal{R}$-structures
even when the individual probabilities of relational atoms are given
by $q_{R}$ rather than $\frac{1}{2}$:
\begin{fact}\label{fact:asymptotic_theory_FOL}
  $\underset{n\rightarrow\infty}{\lim}Q_{T}^{(n)}(\varphi)=1$
for all abstract distributions $T$ over $\mathcal{R}$ and all extension
axioms $\varphi$  over $\mathcal{R}$.
\end{fact}
\begin{cor}
  First-order logic has asymptotic quantifier elimination.
\end{cor}

\section{Probabilistic logic programs as least fixed point distributions}\label{sec:LFP}

We will now proceed briefly to discuss fixed point logics. Our presentation follows the book by Ebbinghaus and Flum (2006, Chapter 8), to which we refer the reader for a more detailed exposition.
We begin by introducing the syntax.

As atomic second-order formulas occur, as subformulas of least fixed point formulas, we will introduce second-order variables. 

\begin{defn}
Assume an infinite set of second-order variables, indicated customarily by upper-case letters from the end of the alphabet, each annotated with a natural number arity. 
Then an \emph{atomic second-order formula} $\varphi$ is either a (first-order) atomic formula, or an expression of the form $X(t_1, \dots, t_n)$, where $X$ is a second-order variable of arity $n$ and $t_1, \dots, t_n$ are constants or (first-order) variables.  
\end{defn}

We now proceed to least fixed point formulas. 

\begin{defn}
\label{def:LFP}A formula $\varphi$ is called\emph{ positive in a
variable $x$} if $x$ is in the scope of an even number of negation
symbols in $\varphi$. 

A \emph{formula in least fixed point logic }or \emph{LFP formula} over a vocabulary
$\mathcal{R}$\emph{ }is defined inductively as follows:
\begin{enumerate}
\item Any atomic second-order formula is an LFP formula.
\item If $\varphi$ is an LFP formula, then so is $\neg\varphi$.
\item If $\varphi$ and $\psi$ are LFP formulas, then so is $\varphi\vee\psi$
\item If $\varphi$ is an LFP formula, then so is $\exists x\varphi$ for
a first-order variable $x$.
\item If $\varphi$ is an LFP formula, then so is $[\mathrm{LFP}_{\vec{x},X}\varphi]\vec{t}$,
where $\text{\ensuremath{\varphi}}$ is positive in the second-order
variable $X$ and the lengths of the string of first-order variables
$\vec{x}$ and the string of terms $\vec{t}$ coincide with the arity
of $X$. 
\end{enumerate}

An occurrence of a second-order variable $X$ is \emph{bound} if it is in the scope of an LFP quantifier $\mathrm{LFP}_{\vec{x},X}$ and \emph{free} otherwise.
\end{defn}

Fixed point semantics have been used extensively in (logic) programming
theory \cite{Fitting02}, and we will exploit this
when relating the model theory of LFP to probabilistic logic programming
below. 

We first associate an operator with each LFP formula $\varphi$:
\begin{defn}
\label{def:F_=00005Cvarphi}Let $\varphi(\vec{x},\vec{u},X,\vec{Y})$
be an LFP formula, with the length of $\vec{x}$ equal to the arity
of $X$, and let $\omega$ be an $\mathcal{R}$-structure with domain
$D$. Let $\vec{b}$ and $\vec{S}$ be an interpretation of $\vec{u}$ and
$\vec{Y}$ respectively. Then we define the operator $F^{\varphi}:\mathfrak{P}(D^{k})\rightarrow\mathfrak{P}(D^{k})$
as follows:
\[
F^{\varphi}(R)\coloneqq\left\{ \vec{a}\in D^{k}|\omega\models\varphi(\vec{a},\vec{b},R,\vec{S})\right\} .
\]
\end{defn}

Since we have restricted Rule 5 in Definition \ref{def:LFP} to positive
formulas, $F^{\varphi}$ is monotone for all $\varphi$ (i. e. $R\subseteq F^{\varphi}(R)$
for all $R\subseteq D^{k}$). Therefore we have:
\begin{fact}
\label{fact:LFP_exists}For every $\mathrm{LFP}$ formula $\varphi(\vec{x},\vec{u},X,\vec{Y})$
and every $\mathcal{R}$-structure on a domain $D$ and interpretation
of variables as in Definition \ref{def:F_=00005Cvarphi}, there is
a relation $R\subseteq D^{k}$ such that $R=F^{\varphi}(R)$ and that
for all $R'$ with $R'=F^{\varphi}(R')$ we have $R\subseteq R'$.
\end{fact}

\begin{defn}
We call the $R$ from Fact \ref{fact:LFP_exists} the \emph{least
fixed point} of $\varphi(\vec{x},\vec{u},X,\vec{Y})$
\end{defn}

Now we are ready to define the semantics of least fixed point logic:
\begin{defn}
By induction on the definition of an LFP formula, we define when
an LFP formula $\varphi(\vec{X},\vec{x})$ is said to \emph{hold}
in an $\mathcal{R}$-structure $\omega$ for a tuple $\vec{a}$ from
the domain of $\omega$ and relations $\vec{A}$ of the correct arity:

The first-order connectives and quantifiers $\neg$, $\vee$ and $\exists$
as well as $\wedge$ and $\forall$ defined from them in the usual
way are given the usual semantics.

An atomic second order formula $X(\vec{x},\vec{c})$ holds if and
only if $(\vec{a},\vec{c_{\omega}})\in A$. 

$[\mathrm{LFP}_{\vec{x},X}\varphi]\vec{t}$ holds if and only if $\vec{a}$
is in the least fixed point of $F^{\varphi(\vec{x},X)}$. 
\end{defn}

\subsection{Probabilistic logic programming}

Our discussion on probabilistic logic programming employs the simplification
proposed by Riguzzi and Swift \shortcite{RiguzziS18} and considers a probabilistic logic program as a stratified Datalog program over probabilistic facts.
This \emph{distribution semantics} covers several different equally expressive formalisms \cite{RiguzziS18,deRaedtK15}.
Note that in particular, probabilistic logic programs as used here do not involve function symbols, unstratified negation or higher-order constructs.

See Appendix B.2 in the supplementary material or the book by Ebbinghaus and Flum (2006, Chapter 9) for an introduction to the syntax and semantics of stratified Datalog programs in line with this paper. 

We will use the notation $(\Pi,P)\vec{t}$ for an intensional symbol
$P$ of a stratified logic program $\Pi$ to mean that ``the program
$\Pi$ proves $P\vec{t}$''. 
\begin{defn}
A \emph{probabilistic logic program }consists of probabilistic facts
and deterministic rules, where the deterministic part is a stratified
Datalog program. We will consider it in our framework of abstract
distribution semantics as follows:

$\mathcal{R}$ is given by relation symbols $R'$ for every probabilistic
fact $p_{R}::R(\vec{x})$, with $q_{R'}\coloneqq p_{R}$. Their arity
is just the arity of $R$.

$\mathcal{S}$ is given by the vocabulary of the probabilistic logic
program and additionally the $R'$ in $\mathcal{R}$.

Let $\Pi$ be the stratified Datalog program obtained by prefixing
the program $\{R'(\vec{x})\leftarrow R(\vec{x})|R'\in\mathcal{R}\}$
to the deterministic rules of the probabilistic logic program.

Then $\phi_{P}$ for a $P\in\mathcal{S}\backslash\mathcal{R}$ is
given by $(\Pi,P)\vec{x}.$ 
\end{defn}

The distribution semantics for probabilistic logic programming is related
to the LFP distribution semantics introduced above through the following
fact \cite[Theorem 9.1.1]{EbbinghausF06}:
\begin{fact}
\label{fact:S-Datalog_to_LFP}For every stratifiable Datalog formula
$(\Pi,P)\vec{x}$ as above, there exists an $\mathrm{LFP}$ formula
$\varphi(\vec{x})$ over the extensional vocabulary $\mathcal{R}$
of $\Pi$ such that for every $\mathcal{R}$-structure $\omega$ and every
tuple $\vec{a}$ of elements of $\omega$ of the same length as $\vec{x}$, $\omega \models\varphi(\vec{a})$ if
and only if $\omega\models(\Pi,P)\vec{a}$.
\end{fact}

\begin{rem*}
In fact, it suffices to consider formulas in the so-called bounded
fixed point logic, whose expressiveness lies between first order logic
and least fixed point logic \cite{EbbinghausF06}. 
\end{rem*}

\begin{notation*}
Although we have allowed second-order variables in the inductive definitions
above, we will assume from now on unless mentioned otherwise that
LFP formulas do not have free second-order variables.
\end{notation*}

\subsection{Asymptotic quantifier elimination for probabilistic logic programming}

We discuss the asymptotic reduction
of LFP to FOL by Blass et al. \shortcite{BlassGK85} and conclude that abstract LFP distributions and therefore probabilistic logic programs have asymptotic quantifier elimination.

The main theorem of Blass et al. \shortcite{BlassGK85} shows that $\mathrm{RANDOM}(\mathcal{R})$
not only eliminates classical quantifiers, but also least fixed point
quantifiers:
\begin{fact}
\label{fact:LFP_FOL}Let $\varphi(\vec{x})$ be an \emph{LFP} formula
over $\mathcal{R}$. Then there is a finite subset $G$ of $\mathrm{RANDOM}(\mathcal{R})$
and a quantifier-free formula $\varphi'(\vec{x})$ such that $G\vdash\forall_{\vec{x}}\varphi(\vec{x})\leftrightarrow\varphi'(\vec{x})$. 
\end{fact}

Putting this together, we can derive the following:
\begin{thm}
\label{thm:QE_LFP_distribution} Least fixed point logic has asymptotic quantifier elimination.
\end{thm}

To obtain a characterisation within probabilistic logic programming,
however, we need to translate quantifier-free first order formulas
back to stratifiable Datalog.

In fact, they can be mapped to a subset of stratified Datalog that
is well-known from logic programming:
\begin{defn}
A Datalog program, Datalog formula or probabilistic logic program
is called \emph{determinate }if every variable occurring in the
body of a clause also occurs in the head of that clause.
\end{defn}

\begin{exam}
  Examples of determinate clauses in this sense are $R(x) \coloneq P(x)$ or $Q(x,y) \coloneq R(x)$.
  Indeterminate clauses include $R(x) \coloneq P(y)$ or $R(x) \coloneq Q(x,y)$.   
\end{exam}

Determinacy corresponds exactly to the fragment of probabilistic
logic programs identified as projective by Jaeger and Schulte (2018, Proposition 4.3).

Indeed, Ebbinghaus and Flum's \shortcite{EbbinghausF06} proof of their Theorem 9.1.1 shows:
\begin{fact}
\label{fact:QF_to_dS-Datalog}Every quantifier-free first order formula
is equivalent to an acyclic determinate stratified Datalog formula.
\end{fact}

Therefore, we can conclude from Proposition \ref{prop:QF_implies_projective}:

\begin{prop}
\label{prop:determinate->projective}Every determinate probabilistic
logic program is projective.
\end{prop}

We now turn to the main result of this
subsection.
\begin{thm}
\label{thm:Probabilistic_Logic_Programs}Every probabilistic logic
program is asymptotically equivalent to an acyclic determinate probabilistic
logic program.
\end{thm}

\section{\label{sec:Discussion} Projective probabilistic logic programs}

As an application of our results, we investigate the projective families of distributions that are expressible by probabilistic logic programs.

The key is the following observation:

\begin{prop}
\label{prop:projective_and_AE_->_Equ}Two projective families of distributions
are asymptotically equivalent if and only if they are equal.
\end{prop}

As modelling in the distribution semantics often involves introducing auxiliary predicates, the family of distributions we want to model will usually be defined on a smaller vocabulary than the abstract distribution (or probabilistic logic program) itself.
We therefore note here that asymptotic equivalence
is preserved under reduct.
First we clarify how we build reducts
of distributions in the first place: 
\begin{defn}
Let $Q_{\text{\ensuremath{}}}^{(n)}$ be a distribution over a vocabulary
$\mathcal{S}$. Then its \emph{reduct} $Q_{\mathcal{S'}}^{(n)}$ \emph{to
a subvocabulary} $\mathcal{S'}\subseteq \mathcal{S}$ is defined such that for
any world $\omega\in\Omega_{n}^{\mathcal{S'}}$, $Q_{\mathcal{S'}}^{(n)}(\omega)\coloneqq Q^{(n)}(\{\omega'\in\Omega_{n}^{\mathcal{S}}|\omega'_{\mathcal{S'}}=\omega\})$. 
\end{defn}

\begin{rem*}
$Q_{\mathcal{T}}^{(n)}$ is the pushforward measure of $Q_{\text{\ensuremath{}}}^{(n)}$
with respect to the reduct projection from $\Omega_{n}^{\mathcal{S}}\rightarrow\Omega_{n}^{\mathcal{T}}$. 
\end{rem*}
We can now formulate preservation of asymptotic equivalence under
reducts:
\begin{prop}
\label{prop:AE-reducts}The reducts of asymptotically equivalent families
of distributions are themselves asymptotically equivalent.
\end{prop}

In combination, we obtain:

\begin{thm}\label{thm:Proj_implies_QF}
  Let $L$ be a logical language with asymptotic quantifier elimination that extends quantifier-free first-order logic. Let $\mathcal{R} \subseteq \mathcal{S}$ be vocabularies, and let $\mathcal{S'} \subseteq \mathcal{S}$. Furthermore let $T$ be an $L$-distribution over $\mathcal{R}$ with vocabulary $\mathcal{S}$. Lastly, let $(Q^{(n)})$ be the family of distributions induced by $T$.
  
  Then the following holds: If $Q_{\mathcal{S'}}^{(n)}$ is projective, then there is a quantifier-free distribution $T_q$ over  $\mathcal{R}$ with vocabulary $\mathcal{S}$ such that $Q_{\mathcal{S'}}^{(n)}$ is the reduct of the family of distributions induced by $T_q$ to $\mathcal{S'}$.
\end{thm}

In particular, a projective family of distributions that can be expressed in probabilistic logic programming at all can in fact be expressed already by a determinate probabilistic logic program. 

\section{Implications and discussion}

The results have immediate consequences for the expressiveness of
probabilistic logic programming. 

We first discuss a particularly striking observation:

\subsection{Asymptotic loss of information}

Very insightful is the case of a probabilistic \emph{rule}, i.e. a
clausal formula annotated with a probability. Because of its intuitive
appeal, this is a widely used syntactic element of probabilistic logic
programming languages such as Problog, and its semantics is defined
by introducing a new probabilistic fact to model the uncertainty of
the rule. More precisely:
\[
p::R(\vec{x}) \coloneq Q_{1}(\vec{x}_{1},\vec{y}_{1}),\ldots,Q_{n}(\vec{x}_{n},\vec{y}_{n})
\]
(where $\vec{x}$ are the variables appearing in $R$, $\vec{x_{i}}\subseteq\vec{x}$)
is interpreted as
\[
p::I(\vec{x},\vec{y});R(\vec{x}) \coloneq Q_{1}(\vec{x}_{1},\vec{y}_{1}),\ldots,Q_{n}(\vec{x}_{n},\vec{y}_{n}),I(\vec{x},\vec{y})
\]
(where $\vec{y}\coloneqq\bigcup\vec{y}_{i}$).

It is now easy to see from Proposition \ref{prop:form_QE} that in
the asymptotic quantifier-free representation of this probabilistic
rule, $I$ will no longer occur, since it originally occurred implicitly
quantified in the body of the clause. However, $I$ was the only connection
between the probability annotation of the rule and its semantics!
Therefore, the asymptotic probability of $R(\vec{x})$ is independent
of the probability assigned to any non-determinate rule with
$R(\vec{x})$ as its head. 

\subsection{Expressing projective families of distributions}

Our results also show how few of the projective families of distributions
can be expressed in those formalisms. This confirms the suspicion
voiced in by Jaeger and Schulte \shortcite{JaegerS20} that despite the ostensible similarities
between languages such as independent choice logic, which are based
on the distribution semantics, and the array representation introduced by Jaeger and Schulte \shortcite{JaegerS20}, a direct application of techniques from probabilistic logic programming to general projective families of distributions might prove challenging.

We start by displaying some properties shared by the projective distributions induced by a probabilistic logic program.

\begin{defn}
  A projective family of distributions has the \emph{Independence Property} or \emph{IP} if for all $\mathcal{S}$-formulas $\varphi(x_1, \dots x_n)$ and $\psi(x_1, \dots x_m)$  the events $\{1, \dots, n\} \models \varphi$ and $\{n+1, \dots, n+m\} \models \psi $ are independent under $Q^{(n+m)}$.
  A projective family of distributions $(Q^{(n)})$ of $\mathcal{S}$-structures  has the \emph{Conditional Independence-Property} or \emph{CIP} if for all $n$ and all quantifier-free $\mathcal{S}$-formulas $\varphi(x_1, \dots x_n)$ and every $\mathcal{S}$-structure $\omega$ on a domain with $n-1$ elements, the events  $\{1, \dots, n\} \models \varphi $ and $\{1, \dots, n+1\} \setminus \{n\} \models \varphi $ are conditionally independent over $\{1, \dots, n-1\} \models \omega$ under $Q^{(n+1)}$.
\end{defn}

IP has been studied extensively in the field of pure inductive logic \cite{ParisV15}, while CIP is a generalisation of the property that Jaeger and Schulte \shortcite{JaegerS20} claimed in their Proposition 7.1 for AHK- distributions, to arbitrary quantifier-free formulas rather than worlds.

\begin{exam}\label{exam:CIP-IP}
  Consider the quantifier-free abstract distribution with a probabilistic fact $R(x)$ with associated probability $p$ and a binary predicate $P(x,y)$ with definition $\phi_P$ = $x = y \vee R(x)$. Then its induced family of distributions satisfies CIP and IP. 
  However, the reduct to the vocabulary $\{P\}$ does not satisfy CIP; indeed, consider the domain with elements $\{1, 2, 3\}$. Then there is just one $\{P\}$-structure $\omega$ with domain $\{1\}$ that has probability 1, namely the world where $P(1,1)$ is true. Consider the events $P(1,3)$ and $P(1,2)$. They are not independent, since in fact $P(1,2)$ iff $R(1)$ iff $P(1,3)$. Since there is just one possible $\{P\}$-structure  $\omega$ on $\{1\}$, conditioning on $\omega$ does not alter the probabilities.  
\end{exam}

\begin{prop} \label{prop:CIP-IP}
  Let $(Q^{(n)})$ be a projective family of distributions induced by a quantifier-free abstract distribution. Then  $(Q^{(n)})$ satisfies CIP. If it does not have any nullary relation symbols, it also satisfies IP.
\end{prop}

As mentioned above, one often expands the vocabulary of interest when modelling in the distribution semantics. It is worth noting, therefore, that IP is trivially transferred to reducts, while CIP is not (see Example \ref{exam:CIP-IP} above). 
We can view our results as positive or negative, depending on our viewpoint. We will begin with the positive formulation:

\begin{cor} \label{cor:Proj_CIP}
  If a projective family of distributions is induced by a probabilistic logic program, it satisifies CIP.
\end{cor}

As CIP is a generalisation of the property claimed by Jaeger and Schulte \shortcite{JaegerS20} in their Proposition 7.1, this shows that while the class of AHK- representations does not satisfy this property (see the discussion in the appendix to Jaeger and Schulte's corrected version \shortcite{JaegerS20a}), every projective family of distributions induced by a probabilistic logic program does.

Since CIP does not transfer to reducts, however, we look towards IP for a property that all projective families of distributions expressible in probabilistic logic programming satisfy.

\begin{cor}\label{cor:Proj_IP}
  Let $\mathcal{S'} \subseteq \mathcal{S}$ be relational vocabularies without nullary relation symbols. Then for every probabilistic
logic program with vocabulary $\mathcal{S}$, if the reduct $(Q^{(n)}_{\mathcal{S'}})$  is projective,  $(Q^{(n)}_{\mathcal{S'}})$ satisfies IP.
\end{cor}

If we allow nullary relations in $\mathcal{S}$, we obtain finite sums of distributions with IP instead.

\begin{prop}\label{prop:finite_sums_of_IP}
  Let $\mathcal{S'} \subseteq \mathcal{S}$ be relational vocabularies, possibly with nullary relation symbols. Then for every probabilistic
logic program with vocabulary $\mathcal{S}$, if the reduct $(Q^{(n)}_{\mathcal{S'}})$  is projective,  $(Q^{(n)}_{\mathcal{S'}})$ is a finite sum of distributions satisfying IP. 
\end{prop}

It is natural to ask how strong the condition imposed by the previous results is, bringing us to the negative part of our results.
As a special case, we consider a unary vocabulary $\mathcal{S'}$, which only has unary relation symbols, since the projective families of distributions are very well understood for such vocabularies.

Here, \emph{de Finetti's Representation Theorem} \cite[Chapter 9]{ParisV15} says that the projective families of distributions in a unary vocabulary are precisely the \emph{potentially infinite} combinations of those that satisfy IP, while those projective families of distributions expressible in probabilistic logic programs are merely the \emph{finite} combinations of those satisfying IP; so, in some sense ``almost all'' projective families of distributions in unary vocabularies cannot be expressed in probabilistic logic programming. 

As a concrete example, we show that already in the very limited vocabulary of a
single unary relation symbol $R$, there is no probabilistic logic
program that induces the distribution that is uniform on isomorphism
classes of structures:
\begin{defn}
\label{def:Carnap's function}Let $\mathcal{S}\coloneqq\{R\}$ consist
of one unary predicate, and let $\mathfrak{m}^{*}$ be the family
of distributions on $\mathcal{S}$-structures defined by 

$\mathfrak{m}^{*}(\{\omega\})\coloneqq\frac{1}{|D|*N_{\omega}}$ for
a world $\omega\in\Omega_{D}$, where $N_{\omega}\coloneqq\left|\left\{ \omega'\in\Omega_{D}|\omega\cong\omega'\right\} \right|$.

This gives each isomorphism type of structures equal weight, and then
within each isomorphism type every world is given equal weight too.
\end{defn}

$\mathfrak{m}^{*}$ is an important probability measure for two reasons;
it plays a special role in finite model theory since the so-called
unlabeled 0-1 laws are introduced with respect to this measure. Furthermore,
it was introduced explicitly by Carnap \shortcite{Carnap50,Carnap52} as
a candidate measure for formalising inductive reasoning, as part of
the so-called \emph{continuum of inductive methods}.
Paris and Vencovsk\'{a} \shortcite{ParisV15} provide a modern exposition of Carnap's theory. 

$\mathfrak{m}^{*}$ is easily seen to be exchangeable; it is also projective, and in
fact an elementary calculation shows that for any domain $D$ and
any $\left\{ a_{1},\ldots a_{n+1}\right\} \subseteq D$, 
\begin{equation}
\mathfrak{m}^{*}\left(R(a_{n+1})|\left\{ R(a_{i})\right\} _{i\in I\subseteq\{1,\ldots,n\}}\cup\left\{ \neg R(a_{i})\right\} _{i\in\{1,\ldots,n\}\backslash I}\right)=\frac{\left|I\right|+1}{n+1}\label{eq:m*}
\end{equation}
(see any of the sources above for a derivation). 
\begin{prop}\label{prop:Carnap_function}
Let $\mathcal{S}'$ be a finite vocabulary extending $\mathcal{S}$
from Definition \ref{def:Carnap's function}. Then there is no probabilistic
logic program with vocabulary $\mathcal{S}'$ such that the reduct
of the induced family of distributions to $\mathcal{S}$ is equal
to $\mathfrak{m}^{*}$.

\end{prop}

\subsection{Complexity results}
Since the theory of random structures is decidable, the asymptotic quantifier results in this paper provide us with an algorithmic procedure for determining an asymptotically equivalent acyclic determinate program for any given probabilistic logic program, and by extension a procedure for determining the asymptotic probabilities of quantifier-free queries.   
What can we say about the complexity of this procedure?
Since the operation takes a non-ground probabilistic logic
program as input and computes another probabilistic logic program,
the notion of data complexity does not make sense in this context.
Instead, program complexity is an appropriate measure. 

In our context, the input program could be measured in different ways.
Since our analysis is based on the setting of abstract distributions,
we will be considering as our input abstract distributions obtained
from (stratified) probabilistic logic programs . We will furthermore
fix our vocabularies $\mathcal{R}$ and $\mathcal{S}$. Since the transformation
acts on each $\phi_{R}$ in turn and independently, it suffices to
consider the individual $\phi_{R}$ as input. It is natural to ask
about complexity in the \emph{length} of $\phi_{R}$.

In fact, one can extract upper and lower bounds from the work of Blass et al. \shortcite{BlassGK85}, who build on the work of Grandjean \shortcite{Grandjean83} for analysing the complexity of their asymptotic results.
The task of determining whether the probability
of a first-order sentence converges to 0 or 1 with increasing domain
size, which is a special case of our transformation, is complete in
PSPACE \cite[Theorem 1.4]{BlassGK85}. Therefore the program transformation
is certainly PSPACE-hard. On the other hand, asymptotic elimination
of quantifiers in least fixed point logic is complete in EXPTIME \cite[Theorems 4.1 and 4.3]{BlassGK85}, so the program transformation is certainly
in EXPTIME. 

In order to specify further, we note that for abstract first-order
distributions, which correspond to acyclic probabilistic logic programs,
the transformation can be performed in PSPACE:

Let $R$ be of arity $n$. Then enumerate the (finitely many) quantifier-free
$n$-types $\left(\varphi_{i}\right)$ in $\mathcal{R}$. Now for
any $\phi_{R}$of arity $n$ we can check successively in polynomial
space in the length of $\phi_{R}$, whether the probability of $\varphi_{i}\rightarrow\phi_{R}$converges
to 0 or 1. Then $\phi_{R}$ is asymptotically equivalent to the conjunction
of those quantifier-free $n$-types for which 1 is returned. 

In the general case of least fixed point logic, Blass et al. \shortcite{BlassGK85}
show that the problem of finding an asymptotically equivalent first-order
sentence is EXPTIME complete. However, to represent stratified Datalog,
only the fragment known as \emph{bounded }or \emph{stratified }least
fixed point logic is required \cite[Sections 8.7 and 9.1]{EbbinghausF06}.
Therefore, the complexity class of the program transformation of stratified
probabilistic logic programs corresponds to the complexity of the
asymptotic theory of \emph{bounded} fixed point logic, which to the
best of our knowledge is still open.

\section{Conclusion and further Work}

By introducing the formalism of abstract distributions, we have related the asymptotic analysis of finite model theory to the distribution semantics underlying probabilistic logic programming.
Thereby, we have shown that every probabilistic logic program is asymptotically equivalent to an acyclic determinate logic program.
In particular, this representation provides us with an algorithm to evaluate the asymptotic probabilities of quantifier-free queries with respect to a probabilistic logic program.
We have also seen that the asymptotic representation of a probabilistic logic program invoking probabilistic rules is in fact independent of the probability with which the rule is annotated.
We applied our asymptotic results to study the projective families of distributions that can be expressed in probabilistic logic programming. 
We saw that they have certain independence properties, and that in particular the families of distributions induced on the entire vocabulary satisfy the conditional independence property.
We also see that at least in the case of a unary vocabulary, only a minority of projective families of distributions can be represented, excluding important example such as Carnap's family of distributions $\mathfrak{m}^*$.  

\subsection{Further work}

The analysis presented here suggests several strands of further research. 

While some widely used directed frameworks can be subsumed under the
probabilistic logic programming paradigm, undirected models such as Markov
logic networks (MLNs) seem to require a different approach.
Indeed, the projective fragment of MLNs isolated by Jaeger and Schulte \shortcite{JaegerS18} is particularly
restrictive, since it only allows formulas in which every literal
has the same variables. Those are precisely the $\sigma$-determinate formulas discussed by Domingos and Singla \shortcite{DomingosS07};
cf. also the parametric classes of finite model theory \cite[Section 4.2]{EbbinghausF06}. It might therefore be expected
that if an analogous result to Theorem \ref{thm:Probabilistic_Logic_Programs}
holds for MLNs, they could express even fewer projective families
of distributions than probabilistic logic programs.

Beyond the FOL or LFP expressions used in current probabilistic logic
programming, another direction is to explore languages with more expressive
power. Candidates for this are for instance Keisler's \shortcite{Keisler85} logic with probability
quantifiers or Koponen's \shortcite{Koponen20} conditional probability
logic. Appropriate asymptotic quantifier elimination
results have been shown in both settings \cite{Koponen20,KeislerL09},
allowing an immediate application of our results there.

The asymptotic quantifier elimination presented here excludes 
higher-order programming constructs from our probabilistic logic programs.
Investigating the asymptotic theory of impredicative programs under a formalised semantics
such as that presented by Bry \shortcite{Bry20} could have direct consequences for the expressiveness of such more general probabilistic logic programs.   

Finally, the failure of the classical paradigm under investigation
to express general projective families of distributions suggests one
may have to look beyond the current methods and statistical relational frameworks
to address the challenge of learning and inference for general projective
families of distributions issued by Jaeger and Schulte \shortcite{JaegerS20}.

\bibliographystyle{acmtrans}
\bibliography{Probabilisticlogicbib}

\begin{thebibliography}{}

\bibitem[\protect\citeauthoryear{Blass, Gurevich, and Kozen}{Blass
  et~al\mbox{.}}{1985}]{BlassGK85}
{\sc Blass, A.}, {\sc Gurevich, Y.}, {\sc and} {\sc Kozen, D.} 1985.
\newblock A zero-one law for logic with a fixed-point operator.
\newblock {\em Inf. Control.\/}~{\em 67,\/}~1-3, 70--90.

\bibitem[\protect\citeauthoryear{Bry}{Bry}{2020}]{Bry20}
{\sc Bry, F.} 2020.
\newblock In praise of impredicativity: {A} contribution to the formalization
  of meta-programming.
\newblock {\em Theory Pract. Log. Program.\/}~{\em 20,\/}~1, 99--146.

\bibitem[\protect\citeauthoryear{Carnap}{Carnap}{1950}]{Carnap50}
{\sc Carnap, R.} 1950.
\newblock {\em Logical Foundations of Probability}.
\newblock University of Chicago Press.

\bibitem[\protect\citeauthoryear{Carnap}{Carnap}{1952}]{Carnap52}
{\sc Carnap, R.} 1952.
\newblock {\em The Continuum of Inductive Methods}.
\newblock University of Chicago Press.

\bibitem[\protect\citeauthoryear{Cozman and Mau{\'{a}}}{Cozman and
  Mau{\'{a}}}{2019}]{CozmanM19}
{\sc Cozman, F.~G.} {\sc and} {\sc Mau{\'{a}}, D.~D.} 2019.
\newblock The finite model theory of bayesian network specifications:
  Descriptive complexity and zero/one laws.
\newblock {\em Int. J. Approx. Reason.\/}~{\em 110}, 107--126.

\bibitem[\protect\citeauthoryear{Domingos and Singla}{Domingos and
  Singla}{2007}]{DomingosS07}
{\sc Domingos, P.~M.} {\sc and} {\sc Singla, P.} 2007.
\newblock Markov logic in infinite domains.
\newblock In {\em Probabilistic, Logical and Relational Learning - {A} Further
  Synthesis, 15.04. - 20.04.2007}, {L.~D. Raedt}, {T.~G. Dietterich},
  {L.~Getoor}, {K.~Kersting}, {and} {S.~Muggleton}, Eds. Dagstuhl Seminar
  Proceedings, vol. 07161. Internationales Begegnungs- und Forschungszentrum
  fuer Informatik (IBFI), Schloss Dagstuhl, Germany.

\bibitem[\protect\citeauthoryear{Ebbinghaus and Flum}{Ebbinghaus and
  Flum}{2006}]{EbbinghausF06}
{\sc Ebbinghaus, H.} {\sc and} {\sc Flum, J.} 2006.
\newblock {\em Finite model theory, Second Edition}.
\newblock Springer Monographs in Mathematics. Springer.

\bibitem[\protect\citeauthoryear{Fitting}{Fitting}{2002}]{Fitting02}
{\sc Fitting, M.} 2002.
\newblock Fixpoint semantics for logic programming a survey.
\newblock {\em Theor. Comput. Sci.\/}~{\em 278,\/}~1-2, 25--51.

\bibitem[\protect\citeauthoryear{Grandjean}{Grandjean}{1983}]{Grandjean83}
{\sc Grandjean, E.} 1983.
\newblock Complexity of the first-order theory of almost all finite structures.
\newblock {\em Inf. Control.\/}~{\em 57,\/}~2/3, 180--204.

\bibitem[\protect\citeauthoryear{Jaeger and Schulte}{Jaeger and
  Schulte}{2018}]{JaegerS18}
{\sc Jaeger, M.} {\sc and} {\sc Schulte, O.} 2018.
\newblock Inference, learning, and population size: Projectivity for {SRL}
  models.
\newblock In {\em Eighth International Workshop on Statistical Relational AI
  (StarAI)}.

\bibitem[\protect\citeauthoryear{Jaeger and Schulte}{Jaeger and
  Schulte}{2020a}]{JaegerS20}
{\sc Jaeger, M.} {\sc and} {\sc Schulte, O.} 2020a.
\newblock A complete characterization of projectivity for statistical
  relational models.
\newblock In {\em Proceedings of the Twenty-Ninth International Joint
  Conference on Artificial Intelligence, {IJCAI} 2020}, {C.~Bessiere}, Ed.
  ijcai.org, 4283--4290.

\bibitem[\protect\citeauthoryear{Jaeger and Schulte}{Jaeger and
  Schulte}{2020b}]{JaegerS20a}
{\sc Jaeger, M.} {\sc and} {\sc Schulte, O.} 2020b.
\newblock A complete characterization of projectivity for statistical
  relational models, version 2.
\newblock {\em CoRR\/}~{\em abs/2004.10984v2}.

\bibitem[\protect\citeauthoryear{Keisler}{Keisler}{1985}]{Keisler85}
{\sc Keisler, H.~J.} 1985.
\newblock Probability quantifiers.
\newblock In {\em Model-theoretic logics}. Perspect. Math. Logic. Springer, New
  York, 509--556.

\bibitem[\protect\citeauthoryear{Keisler and Lotfallah}{Keisler and
  Lotfallah}{2009}]{KeislerL09}
{\sc Keisler, H.~J.} {\sc and} {\sc Lotfallah, W.~B.} 2009.
\newblock Almost everywhere elimination of probability quantifiers.
\newblock {\em J. Symb. Log.\/}~{\em 74,\/}~4, 1121--1142.

\bibitem[\protect\citeauthoryear{Koponen}{Koponen}{2020}]{Koponen20}
{\sc Koponen, V.} 2020.
\newblock Conditional probability logic, lifted bayesian networks, and almost
  sure quantifier elimination.
\newblock {\em Theor. Comput. Sci.\/}~{\em 848}, 1--27.

\bibitem[\protect\citeauthoryear{Paris and Vencovsk\'{a}}{Paris and
  Vencovsk\'{a}}{2015}]{ParisV15}
{\sc Paris, J.} {\sc and} {\sc Vencovsk\'{a}, A.} 2015.
\newblock {\em Pure inductive logic}.
\newblock Perspectives in Logic. Association for Symbolic Logic, Ithaca, NY;
  Cambridge University Press, Cambridge.

\bibitem[\protect\citeauthoryear{Poole, Buchman, Kazemi, Kersting, and
  Natarajan}{Poole et~al\mbox{.}}{2014}]{PooleBKKN14}
{\sc Poole, D.}, {\sc Buchman, D.}, {\sc Kazemi, S.~M.}, {\sc Kersting, K.},
  {\sc and} {\sc Natarajan, S.} 2014.
\newblock Population size extrapolation in relational probabilistic modelling.
\newblock In {\em Scalable Uncertainty Management - 8th International
  Conference, {SUM} 2014, Oxford, UK, September 15-17, 2014. Proceedings},
  {U.~Straccia} {and} {A.~Cal{\`{\i}}}, Eds. Lecture Notes in Computer Science,
  vol. 8720. Springer, 292--305.

\bibitem[\protect\citeauthoryear{Raedt and Kimmig}{Raedt and
  Kimmig}{2015}]{deRaedtK15}
{\sc Raedt, L.~D.} {\sc and} {\sc Kimmig, A.} 2015.
\newblock Probabilistic (logic) programming concepts.
\newblock {\em Mach. Learn.\/}~{\em 100,\/}~1, 5--47.

\bibitem[\protect\citeauthoryear{Riguzzi and Swift}{Riguzzi and
  Swift}{2018}]{RiguzziS18}
{\sc Riguzzi, F.} {\sc and} {\sc Swift, T.} 2018.
\newblock A survey of probabilistic logic programming.
\newblock In {\em Declarative Logic Programming: Theory, Systems, and
  Applications}, {M.~Kifer} {and} {Y.~A. Liu}, Eds. {ACM} / Morgan {\&}
  Claypool, 185--228.

\end{thebibliography}

\appendix
\section{Proofs}

In the appendix we collate the proofs of the claims made in the paper.

\begin{proof}[Proof of Proposition \ref{prop:QF_implies_projective}]
  We have seen that the independent distribution induced on the space of $\mathcal{R}$-structures with domain $D$ is projective. Additionally, closed quantifier-free formulas hold in a substructure if and only if they hold in the original structure. So let $\omega$ be an $\mathcal{S}$-structure with domain $D$ and let $D \subset D'$. Let $\omega_{\mathcal(R)}$ be the $\mathcal{R}$-structure on $\omega$.
  If $\omega$ has probability $0$ because $\omega \models \exists_{\vec{a}\in\vec{D}}\exists_{R\in\mathcal{S} \setminus \mathcal{R}}:R(\vec{a})\nLeftrightarrow\phi_{R}(\vec{a})$, then so will all superstructures of $\omega$ since existential formulas are closed under superstructure. So assume this is not the case. 
  Then
\[
  Q_T^{D'}\left(\left\{ \omega_{\mathcal(R)}'\in\Omega_{n}|\textrm{\ensuremath{\omega_{\mathcal(R)}} is the substructure of }\omega_{\mathcal(R)}'\textrm{ with domain \ensuremath{\left\{  1,\ldots,n\right\} } }\right\} \right) = Q_T^{D}(\omega_{\mathcal(R)}).
\]

  Every $\omega_{\mathcal(R)}'$ has a unique extension to an $\mathcal{S}$-structure $\omega'$ with  $R(\vec{a})\nLeftrightarrow\phi_{R}(\vec{a})$ for all $R \in \mathcal{S} \setminus \mathcal{R}$. Since quantifier-free formulas with values in $\omega_{\mathcal(R)}$ are true in $\omega_{\mathcal(R)}'$ if and only if they are true in $\omega_{\mathcal(R)}$, those are exactly the extensions of $\omega$ to $D'$ that have non-zero weight.  
\end{proof}

\begin{proof}[Proof of Proposition  \ref{prop:form_QE}]
  The first claim follows from the fact that $\mathrm{RANDOM}(\mathcal{T})$ is the reduct of  $\mathrm{RANDOM}(\mathcal{R})$ to $\mathcal{T}$ for any $\mathcal{T}\subseteq\mathcal{R}$.

To show the second claim, consider the vocabulary $\bar{\mathcal{R}}$
containing $R_{\vec{x}}(\vec{y})$ for every atomic subformula $R(\vec{x},\vec{y})$
of $\varphi$ and let $\bar{\varphi}$ be the $\bar{\mathcal{R}}$-formula
obtained from $\varphi$ by replacing every occurrence of $R(\vec{x},\vec{y})$
with $R_{\vec{x}}(\vec{y})$. Let $M$ be a model of $\mathrm{\ensuremath{RANDOM}(\mathcal{R})}$
and let $\vec{a}\in M$. Then define an $\bar{\mathcal{R}}$-structure
on $M$ by setting $R_{\vec{x}}(\vec{y})\colon\Leftrightarrow R(\vec{a},\vec{y})$.
One can verify that $M$ satisfies the extension axioms in $\mathrm{RANDOM}(\bar{\mathcal{R}})$.
Since $\mathrm{RANDOM}(\bar{\mathcal{R}})$ is complete, $\mathrm{\ensuremath{RANDOM}(\bar{\mathcal{R}})}\vdash\bar{\varphi}$
or $\mathrm{\ensuremath{RANDOM}(\bar{\mathcal{R}})}\vdash\neg\bar{\varphi}$.
Therefore, either $\varphi(\vec{a})$ or $\neg\varphi(\vec{a})$ holds
uniformly for all $\vec{a}\in M$. Therefore, either $\mathrm{\ensuremath{RANDOM}(\mathcal{R})}\vdash\forall_{\vec{x}}\varphi(\vec{x})$
or $\mathrm{\ensuremath{RANDOM}(\mathcal{R})}\vdash\forall_{\vec{x}}\neg\varphi(\vec{x})$. 
\end{proof}

\begin{proof}[Proof of Theorem \ref{thm:QE_LFP_distribution}]
By Fact \ref{fact:LFP_FOL} and the finiteness of the vocabulary $\mathcal{V}$,
there is a finite set $G$ of extensions axioms over $\mathcal{R}$
such that there are quantifier-free $\mathcal{R}$-formulas $\phi'_{R}$
for every $R\in\mathcal{V}\backslash\mathcal{R}$ with $G\vdash\forall_{\vec{x}}\phi_{R}(\vec{x})\leftrightarrow\phi_{R}'(\vec{x})$. 

By Fact \ref{fact:asymptotic_theory_FOL}, $\underset{n\rightarrow\infty}{\lim}Q_{T}^{(n)}(\{\omega\in\Omega_{n}|\omega_{\mathcal{R}}\models G\})=1$
for any finite subset $G\subseteq\mathrm{RANDOM}(\mathcal{R})$ and
thus $\underset{n\rightarrow\infty}{\lim}Q_{T}^{(n)}(\{\omega\in\Omega_{n}|\forall_{R\in\mathcal{V}\backslash\mathcal{R}}\omega_{\mathcal{R}}\models\phi_{R}\leftrightarrow\phi'_{R}\})=1$.
Let $(Q_{T}^{(n)})$ be the family of distributions induced by the
quantifier-free $\mathrm{\mathrel{FO}}$-distribution over $\mathcal{R}$,
in which every $\phi_{R}$ is replaced by $\phi'_{R}$. By construction,
$Q^{(n)}(\omega)=Q'^{(n)}(\omega)$ for every world $\omega$ with
$\forall_{R\in\mathcal{V}\backslash\mathcal{R}}\omega_{\mathcal{R}}\models\phi_{R}\leftrightarrow\phi'_{R}$.
Therefore, $\underset{A\subseteq\Omega_{n}}{\sup}|Q^{(n)}(A)-Q'^{(n)}(A)|$
is bounded by above by $1-Q^{(n)}(\{\omega\in\Omega_{n}|\forall_{R\in\mathcal{V}\backslash\mathcal{R}}\omega_{\mathcal{R}}\models\phi_{R}\leftrightarrow\phi'_{R}\})$,
which limits to 0 since $\underset{n\rightarrow\infty}{\lim}Q^{(n)}(\{\omega\in\Omega_{n}|\forall_{R\in\mathcal{V}\backslash\mathcal{R}}\omega_{\mathcal{R}}\models\phi_{R}\leftrightarrow\phi'_{R}\})=1$.
\end{proof}

\begin{proof}[Proof of Theorem \ref{thm:Probabilistic_Logic_Programs}]
Let $\mathcal{S}\backslash\mathcal{R}$ be the extensional vocabulary
of the probabilistic logic program $\varTheta$ and let $\Pi$ be
its underlying Datalog program. Then for every relation $R\in\mathcal{S\backslash\mathcal{R}}$,
$R(\vec{t})$ is given by the Datalog formula $(\Pi,R)\vec{t}$ over
any given $\mathcal{R}$-structure. By Fact \ref{fact:S-Datalog_to_LFP},
$(\Pi,R)\vec{t}$ is equivalent to an LFP formula $\phi_{R}$ over
$\mathcal{R}$. Let $T$ be the abstract LFP distribution over $\mathcal{R}$
in which for every $R\in\mathcal{R}$, $q_{R}$ is taken from $\varTheta$
and for every $R\in\mathcal{S}\backslash\mathcal{R}$, this $\phi_{R}$
is used. Then $T$ and $\varTheta$ induce equivalent families of
distributions. By Theorem \ref{thm:QE_LFP_distribution}, $T$ is
asymptotically equivalent to a quantifier-free abstract distribution,
which in turn is equivalent to a determinate probabilistic
logic program by Fact \ref{fact:QF_to_dS-Datalog}. Therefore $\varTheta$
itself is asymptotically equivalent to a determinate probabilistic
logic program. 
\end{proof}

\begin{proof}[Proof of Proposition \ref{prop:projective_and_AE_->_Equ}]
We will proceed by contradiction. So assume not. Then there is an
$m$ such that $Q^{(m)}$ and $Q'^{(m)}$ are not equal. Let $\omega_{0}$
be a world of size $m$ which does not have the same probability in
$Q^{(m)}$ and $Q'^{(m)}$. Let $a\coloneqq|Q^{(m)}(\{\omega_{0}\})-Q'^{(m)}(\{\omega_{0}\})|$
For any $n\geq m$, consider the subset
\[
A_{n}\coloneqq\{\omega\in\Omega_{n}|\textrm{the substructure of } \omega \textrm{ on the domain } \{1, \dots, m\}\textrm{ is }\omega_{0}\}.
\]
Since both families are projective, $|Q^{(n)}(A_{n})-Q'^{(n)}(A_{n})|=|Q^{(m)}(\{\omega_{0}\})-Q'^{(m)}(\{\omega_{0}\})|=a$.
Therefore, $(Q^{(n)})$ and $(Q'^{(n)})$ are not asymptotically equivalent. 
\end{proof}

\begin{proof}[Proof of Proposition \ref{prop:AE-reducts}]
Let $(Q^{(n)})$ and $(Q^{\prime(n)})$ be asymptotically equivalent
families of distributions over $\mathcal{S}$. Then for any $\mathcal{T}\subseteq\mathcal{S}$,
and any $A\subseteq\Omega_{n}$, 
\[
|Q_{\mathcal{T}}^{(n)}(A)-Q_{\mathcal{\mathcal{T}}}^{\prime(n)}(A)|=|Q_{\mathcal{}}^{(n)}(\{\omega\in\Omega_{n}^{\mathcal{S}}|\omega_{\mathcal{T}}\in A\})-Q_{\mathcal{}}^{\prime(n)}(\{\omega\in\Omega_{n}^{\mathcal{S}}|\omega_{\mathcal{T}}\in A\})|.
\]
 Therefore, $\stackrel[n\rightarrow\infty]{}{\lim}\underset{A\subseteq\Omega_{n}^{\mathcal{T}}}{\sup}|Q_{\mathcal{T}}^{(n)}(A)-Q_{\mathcal{\mathcal{T}}}^{\prime(n)}(A)|\leq\stackrel[n\rightarrow\infty]{}{\lim}\underset{A\subseteq\Omega_{n}^{\mathcal{S}}}{\sup}|Q^{(n)}(A)-Q'^{(n)}(A)|=0$
\end{proof}

\begin{proof}[Proof of Theorem \ref{thm:Proj_implies_QF}]
  By asymptotic quantifier elimination, we can choose $T_q$ to be asymptotically equivalent to $T$.
  Since $T_q$ is a quantifier-free distribution, its induced family of distributions
  $(P^{(n)})$ is projective.
  By Proposition \ref{prop:AE-reducts}, $Q_{\mathcal{S'}}^{(n)}$ and $P_{\mathcal{S'}}^{(n)}$ are asymptotically equivalent.
  However, since they are both projective, this implies that they are actually equivalent everywhere.
\end{proof}

\begin{proof}[Proof of Proposition \ref{prop:CIP-IP}]
  Let the abstract distribution be defined over $\mathcal{R}$. Then by replacing occurrences of other relations with their quantifier-free definitions, we can reduce to the case where all formulas and structures mentioned are $\mathcal{R}$-formulas and $\mathcal{R}$-structures.
  Since by definition of the abstract distribution semantics, $\{1, \dots, n\} \models \varphi$ and $\{n+1, \dots, n+m\} \models \psi $ are independent for $\mathcal{R}$-structures, this suffices to show IP.

  To show CIP, observe first that for all atoms $R(x_1, \dots, x_k)$ and $n_1, \dots, n_k \in \{1, \dots, n-1\}$, either $\omega \models R(a_1, \dots, a_k)$ or  $\omega \models \neg R(a_1, \dots, a_k)$.
  Therefore, we can replace all occurrences of atoms with entries in $x_1, \dots x_{n-1}$ with $\top$ or $\bot$, depending on whether $\omega$ satisfies them under the substitution $x_i \rightarrow n_i$.
  As only atoms remain in which $x_n$ occurs freely, their interpretations refer to $n$ or $n+1$ in $\{1, \dots, n\}$ and $\{1, \dots, n+1\} \setminus \{n\} $ respectively. Now we can conclude as for IP above.  
\end{proof}

\begin{proof}[Proof of Corollary \ref{cor:Proj_CIP}]
  Such a projective family is in fact induced by a determinate probabilistic logic program, which is equivalent to a quantifier-free family of distributions.
\end{proof}

\begin{proof}[Proof of Corollary \ref{cor:Proj_IP}]
  $(Q^{(n)}_{\mathcal{S'}})$ satisfies IP since $(Q^{(n)})$ does and IP transfers to reducts.
\end{proof}

\begin{proof}[Proof of Proposition \ref{prop:finite_sums_of_IP}]
    As in the previous proofs, we can reduce to the situation where nullary predicates are the propositional facts. 
  Since there are only finitely many nullary predicates in $\mathcal{S}'$,
there are only finitely many possible configurations of those nullary
predicates. For every such configuration $\varphi$, let $q_{\varphi}$ be
the probability of that configuration.
Then the distribution itself is given by the finite sum of the conditional distributions on $\varphi$, weighted by  $q_{\varphi}$ , and every such conditional distribution is given by the probabilistic logic program without nullary relations obtained by substituting $\top$ or $\bot$ for the nullary propositions, depending on whether they are true or false in the configuration $\varphi$.  
\end{proof}

\begin{proof}[Proof of Proposition \ref{prop:Carnap_function}]
Assume there is such a probabilistic logic program. Since $\mathfrak{m}^{*}$ is projective,  it would have to be a finite sum of distributions satisfying IP. For each of these finite components $(P^{(n)})$, let $p_P$ be the unconditional probability of $R(x)$ for any $x$ (well-defined by projectivity).   
We observe from Equation \ref{eq:m*} that for variable $n$, the infimum of $\mathfrak{m}^{*}\left(R(a_{n+1})|\left\{ R(a_{i})\right\} _{i\in I\subseteq\{1,\ldots,n\}}\cup\left\{ \neg R(a_{i})\right\} _{i\in\{1,\ldots,n\}\backslash I}\right)$ is 0, even if we assume that there is at least one $i$ with $R(a_{i})$.
As there are only finitely many components  $(P^{(n)})$, the infimum $c$ of the nonzero $p_P$
is greater than 0. By the IP for the $(P^{(n)})$, $R(a_{n+1})$
is conditionally independent of $\left\{ R(a_{i})\right\} _{i\in I\subseteq\{1,\ldots,n\}}\cup\left\{ \neg R(a_{i})\right\} _{i\in\{1,\ldots,n\}\backslash I}$
under $(P^{(n)})$. Thus, the conditional probability of $R(a_{n+1})$ given $\left\{ R(a_{i})\right\} _{i\in I\subseteq\{1,\ldots,n\}}\cup\left\{ \neg R(a_{i})\right\} _{i\in\{1,\ldots,n\}\backslash I}$
is a weighted mean of the non-zero $p_P$ and therefore bounded
below by $c>0$, \emph{in contradiction to 0 being the infimum of
}$\mathfrak{m}^{*}\left(R(a_{n+1})|\left\{ R(a_{i})\right\} _{i\in I\subseteq\{1,\ldots,n\}}\cup\left\{ \neg R(a_{i})\right\} _{i\in\{1,\ldots,n\}\backslash I}\right)$.
\end{proof}

\section{Background and notation}

\subsection {First-order logic}
This paper follows the notation of Ebbinghaus and Flum \shortcite{EbbinghausF06}, which we outline here. 
Full information can be found in Chapter 1 there.
A \emph{vocabulary}, sometimes called a \emph{relational} vocabulary for emphasis, is a finite set of relation symbols, each of which are assigned a natural number arity, and of constant symbols, but does not contain function symbols.
We also assume an infinite set of first-order variables, customarily referred to by lower-case letters from the end of the alphabet, i. e. $u$ to $z$.   
For a vocabulary $\mathcal{S}$, an \emph{atomic $\mathcal{S}$-formula} or \emph{$\mathcal{S}$-atom} is an expression of the form $R(t_1, \dots, t_n)$, where $R$ is a relation symbol of arity $n$ and every $t_i$ is either a variable or a constant.
An \emph{$\mathcal{S}$-literal} is either an atom $\varphi$ or its negation $\neg \varphi$. 
A \emph{quantifier-free $\mathcal{S}$-formula} is a Boolean combination of atoms, where conjunction is indicated by $\wedge$, disjunction by $\vee$ and logical implication by $\rightarrow$. 
We use the big operators $\underset{}{}\bigwedge$
A \emph{first-order} or \emph{FOL-formula} is made up from atoms using Boolean connectives as well as existential and universal quantifiers over variables $x$, indicated by $\exists_x$ and $\forall_x$ respectively.  
To simplify the notation for longer strings of quantifiers, we use the shorthand $\forall_{x_1, \dots, x_n}$ for $\forall_{x_1} \dots \forall_{x_n}$, and analogously for $\exists$. 

In Section \ref{sec:LFP} we also refer to  second-order formulas, which are introduced there. 

Let $\varphi$ be a first-order formula. An occurrence of  a variable $x$ in $\varphi$ is called \emph{bound} if it is in the scope of a quantifier annotated with that variable and \emph{free} otherwise. $x$ is called \emph{free in $\varphi$} if it occurs freely. 
We use the notation $\varphi(x_1, \dots, x_n)$ for $\varphi$  to assert that every free variable in $\varphi$ is from $x_1, \dots, x_n$. 
We also abuse notation and write $R(\vec{x},\vec{c})$ for an atomic formula with constants $\vec{c}$ and free variables $\vec{x}$, even though they don't necessarily appear in that order. 
A formula with no free variables is called a \emph{sentence}, and a set of sentences is called a \emph{theory}. Sentences making up a given theory are also called its \emph{axioms}.

Since tuples such as $x_1, \dots, x_n$ occur frequently and their exact length is often not important, we use the notation $\vec{x}$ to indicate a tuple of arbitrary finite length in many contexts. 

If  $\mathcal{S}$ is a vocabulary, then an \emph{$\mathcal{S}$-structure} $\omega$ consists of a finite non-empty set $D$, the \emph{domain} of $\omega$,
along with an interpretation of the relation symbols and constants of $\mathcal{S}$ as relations and elements of $D$ respectively.
If $R$ is a relation symbol and $c$ a constant, then we $R_{\omega}$ and $c_{\omega}$ for their respective interpretations in $\omega$.
A bijective map $f:D \rightarrow D'$ between the domains of two $\mathcal{S}$-structures $\omega$ and $\omega'$ respectively is an \emph{$\mathcal{S}$-isomorphism} if it maps the interpretation of every relation symbol and constant in $\omega$ to the interpretation of the same symbol in $\omega'$. 
Given a subset $D' \subseteq D$ of the domain of a structure $\omega$, we call a structure $\omega'$ with domain $D'$ the \emph{substructure} of $\omega$ on $D'$ if the interpretation of the relation symbols in $\omega'$ are obtained by restricting the interpretations of the symbols in $\omega$ to $D'$. 
Let $\mathcal{S}' \subseteq \mathcal{S}$ be two vocabularies. 
Then the \emph{reduct} $\omega_{\mathcal{S}'}$ of an  $\mathcal{S}$-structure $\omega$ is given by simply omitting the interpretations of the symbols not in $\mathcal{S}$. In this situation, we call $\omega$ an \emph{extension} of $\omega_{\mathcal{S}'}$. 

Let $\omega$ be an  $\mathcal{S}$-structure, let $\varphi(x_1, \dots, x_n)$ be a first-order $\mathcal{S}$-formula and let $a_1, \dots, a_n$ be a tuple of elements of the domain $D$ of $\omega$.
Then we write $\omega \models \varphi(a_1, \dots, a_n)$ whenever $\varphi(x_1, \dots, x_n)$ \emph{holds} with respect to the interpretation of $a_1, \dots, a_n$ for $x_1, \dots x_n$, and call $\omega$ a \emph{model} of $\varphi(a_1, \dots, a_n)$. 
Similarly, for a theory $T$, we call $\omega$ a \emph{model} of $T$ if $\omega \models \varphi$ for all axioms $\varphi$ in $T$.  
We also express this situation by saying that $\omega$ \emph{satisfies} $\varphi$ or $T$. 

Let $T$ be a theory and $\varphi$ a sentence. We use the notation $T \vdash \varphi$ to indicate that $\omega \models \varphi$ for all models $\omega$ of $T$. 

\subsection{Logic Programming}
Our terminology for logic programs is taken from Chapter 9 of Ebbinghaus and Flum \shortcite{EbbinghausF06}, where one can find a more detailed exposition.

A \emph{general logic program} in a vocabulary $\mathcal{S}$ is a finite set $\Pi$ of clauses of the form $\gamma \leftarrow \gamma_1, \dots, \gamma_n$, where $n \geq 0$, $\gamma$ is an atomic formula and $\gamma_1, \dots, \gamma_n$ are literals. 
We call $\gamma$ the \emph{head} and $\gamma_1, \dots, \gamma_n$ the \emph{body} of the logic program. 
The \emph{intensional} relation symbols are those that occur in the head of any clause of a program, while the relation symbols occurring only in the body of clauses are called \emph{extensional}. 
We write $(\mathcal{S},\Pi)_{\mathrm{ext}}$ for the extensional vocabulary, or $\mathcal{S}_{\mathrm{ext}}$ where the logic program is clear from context. 

An \emph{acyclic} logic program is one in which no intensional relation symbol occurs in the body of any clause (This is at first glance a stronger condition than the usual definition of acyclicity, but by successively unfolding head atoms used in the body of a clause every acyclic logic program in the usual sense can easily be brought to this form). 

A \emph{pure Datalog program} is a general logic program in which no intensional relation symbol appears within any negated literal.
To affix a meaning to a pure Datalog program, consider it as a function which take as input a finite $\mathcal{S}_{\mathrm{ext}}$-structure $\omega$ and returns as output an extension $\omega_{\Pi}$ of $\omega$ to $\mathcal{S}$.
Starting from an empty interpretation of the relation symbols not in $\mathcal{S}_{\mathrm{ext}}$, we successively expand them by applying the rules of $\Pi$.  We give an informal description of this process:

Let $\gamma(\vec{x})$ be the head of a clause and $\gamma_1(\vec{x}, \vec{y}), \dots, \gamma_n(\vec{x}, \vec{y})$ the body, and let $\vec{a}$ be a tuple of elements of the domain $D$ of length equal to $\vec{x}$.
Then whenever there is a tuple $\vec{b}$ such that $\omega \models \gamma_i(\vec{a}, \vec{b})$ for every $i$, we add $\gamma(\vec{a})$ to the interpretation of the relation symbol $R$ of the atom $\gamma$. 
Successively proceed in this manner until nothing can be added by applying any of the clauses of $\Pi$. Since the domain of $\omega$ is finite, this is bound to happen eventually. 

As the restriction for no intensional relation symbol to occur negated in the body of a clause turns out to be too strong for many practical applications, we consider a generalisation to \emph{stratifiable} Datalog programs.
A general logic program $\Pi$ in a vocabulary $\mathcal{S}$ is called a \emph{stratifiable Datalog program} if there is a partition of $\mathcal{S}$ into subvocabularies $\mathcal{S}_{\mathrm{ext}}=\mathcal{S}_0, \dots \mathcal{S}_n$ such that the following holds:

The corresponding logic programs $\Pi_1, \dots, \Pi_n$, where $\Pi_i$ is defined as the set of clauses whose head atom starts with a relation symbol from $\mathcal{S}_i$, are pure Datalog programs, and the extensional vocabulary of $\Pi_i$ is contained in $\mathcal{S}_0 \cup \dots \cup \mathcal{S}_{i-1}$. 

If $\Pi$ is a stratifiable Datalog program in $\mathcal{S}$ and $\omega$ an $\mathcal{S}_{\mathrm{ext}}$-structure, then $\omega$ is obtained by applying $\Pi_1, \dots, \Pi_n$ successively. 

A \emph{(pure/stratified) Datalog formula} is an expression of the form $(\Pi, P)\vec{x}$, where $\Pi$ is a (pure/stratified) Datalog program, $P$ an intensional relation symbol and $\vec{x}$ a tuple of variables. 
We say that a Datalog formula \emph{holds} in an $\mathcal{S}_{\mathrm{ext}}$-structure $\omega$ for a tuple of elements $\vec{a}$ of the same length as $\vec{x}$, written $\omega \models (\Pi, P)\vec{a}$, if $P(\vec{a})$ is true in $\omega_{\Pi}$. 

\end{document}